\documentclass[journal]{IEEEtran}
\IEEEoverridecommandlockouts
\usepackage{cite}
\usepackage{amsmath,amssymb,amsfonts}
\usepackage{amsthm}
\usepackage{graphicx}
\usepackage{textcomp}
\usepackage{xcolor}
\usepackage{comment}
\usepackage{cite}
\usepackage{algorithm}
\usepackage{xcolor}
\usepackage{hyperref}
\usepackage[capitalise]{cleveref}
\usepackage{svg}
\usepackage{algpseudocodex}
\usepackage{circuitikz}
\usepackage{subcaption}
\usepackage{orcidlink}
\usepackage[disable]{todonotes}
\hypersetup{
    colorlinks,
    linkcolor={red!50!black},
    citecolor={blue!50!black},
    urlcolor={blue!80!black}
}

\Crefname{fig}{Fig.}{Figs.}
\newtheorem{theorem}{Theorem}[section]

\newcommand{\signal}[1]{\texttt{#1}}
\newcommand{\TODO}[1]{\todo[inline]{#1}}
\newcommand{\set}[1]{\left\{ #1 \right\}}

\def\BibTeX{{\rm B\kern-.05em{\sc i\kern-.025em b}\kern-.08em
    T\kern-.1667em\lower.7ex\hbox{E}\kern-.125emX}}
\begin{document}

\title{Latch Based Design for Fast Voltage Droop Response\\

\thanks{This work was supported in part by the European Research Council (ERC) through the European Union’s Horizon 2020 Research and Innovation Programme under Grants 716562 and 101123525.This work was supported in part by the European Research Council (ERC) through the European Union’s Horizon 2020 Research and Innovation Programme under Grants 716562 and 101123525.}}

\author{
\IEEEauthorblockN{1\textsuperscript{st} Shreyas Srinivas}
\IEEEauthorblockA{\textit{CISPA Helmholtz Center for Information Security}
\orcidlink{https://orcid.org/0000-0002-3993-1596}\\
shreyas.srinivas@cispa.de}\\
\and
\IEEEauthorblockN{2\textsuperscript{nd} Ian W. Jones}
\IEEEauthorblockA{\textit{}
Palo Alto, California, USA \\
ian.w.jones@ieee.org}\\
\and
\IEEEauthorblockN{3\textsuperscript{rd} Goran Panic}
\IEEEauthorblockA{\textit{IHP Leibniz-Institute for Innovative Microelectronics} \\
Frankfurt an der Oder, Germany \\
panic@ihp-microelectronics.com}\\
\and
\IEEEauthorblockN{4\textsuperscript{th} Christoph Lenzen}
\IEEEauthorblockA{\textit{CISPA Helmholtz Center for Information Security} \\
Saarbrucken, Germany \\
lenzen@cispa.de}
}

\maketitle

\begin{abstract}
We present a latch-based and PLL-free design of the voltage droop correction circuit of Lenzen, Fuegger, Kinali, and Wiederhake\cite{DroopJournal}. Such a circuit dynamically modifies the clock frequency of a digital clock for VLSI systems. Our circuit responds within two clock cycles and halves the length of the synchroniser chain compared to the previous design. Further, we introduce a differential sensor based design for masking latches as a replacement for masking flip flops that the design of \cite{DroopJournal} requires, but leaves unspecified. The use of latches instead of threshold-altered flip flops alters the timing properties of our design and thus the proofs of correctness that accompanied their design require modifications which we present here. This design has been successfully implemented on the IHP 130 nm process technology. The results of the experimental measurements will be discussed in a subsequent publication.
\end{abstract}

\newcommand{\MCZeroLatch}{\signal{MC 0 Latch}}
\newcommand{\MCZeroOneLatch}{\signal{MC-Latch 01}}
\newcommand{\Latch}{\signal{D Latch}}
\newcommand{\CLKIn}{\signal{CLK\_IN}}
\newcommand{\CLKGated}{\signal{CLK\_Gated}}
\section{Introduction}
Synchronous VLSI circuits are state machines which perform state transitions at each rising clock edge. A critical requirement for maintaining this abstraction is that the combinational circuits computing the transitions finish before the next rising edge of the clock arrives. A droop in the supply voltage $\signal{VDD}$ interferes with this abstraction by slowing down all transistors. A number of solutions exist to deal with this problem, cf.~\cite{DroopJournal} for a discussion. One line of work focuses on adapting the clock frequency on the fly to allow for the extra computation time needed.

A key challenge in this approach lies in the metastability that might arise when measuring the analog property of a voltage droop in the $\signal{VDD}$ signal and converting it into a digital $\signal{droop\_detected}$ value. As noted by Marino~\cite{marino1981general}, metastability cannot be deterministically detected, resolved, or avoided. It arises from the mismatch between the response of physical circuits to continuous time-varying signals and the abstraction of discrete digital values. Previous solutions either rely on analog properties of the circuit, rendering them fragile to process variations and changes in technology, or funnel the droop detection signal through a synchronizer chain first. The latter means that any response to a droop happens several clock cycles after detection of its onset, which is acceptable for slow droops only.

\subsection{Previous Work}
F\"{u}gger, Kinali, and Wiederhake~\cite{DroopJournal} introduce a fast voltage droop response circuit, which responds to the detection of a droop within a clock cycle. They do this by modifying the synchronizer chain to delay the clock signal if the sampled value indicates a droop, while defaulting to fast propagation if metastability resolution is still ongoing. If metastability resolution results in a ``late'' decision that a droop was perceived, the clock signal delay is inserted late as well; as the measurement by the droop detector was inconclusive regarding whether the supply voltage had fallen below acceptable levels, this behavior is safe.

The reason why this design can circumvent Marino's result is that the timing delay can be continuously split between the clock pulse at which the resolution happens and the next one, i.e., no discrete decision between delaying and not delaying clock pulses is forced. The key circuit element enabling this is a \emph{masking flip flop} (hereon referred to as \signal{MFF}), which masks internal metastability until resolution takes place, translating it to a late output transition. In~\cite{DroopJournal}, the implementation of this device is left unspecified, but the following behaviour is required: For $b\in\set{0,1}$, a $\signal{b-MFF}$ is a flip flop that behaves like an ordinary flip flop when there is no metastability. When it is metastable internally, then for the duration of this metastability, it outputs $b$. Once the internal metastability resolves, the output becomes the respective value. Note that if metastability resolves to $b$, metastability is fully masked, whereas resolution to $1-b$ is observed as a late output transition from $b$ to $1-b$. Another critical assumption of their circuit design are stable delays across all used delay lines, which they achieve through the use of PLLs.

\begin{figure*}[t]
    \centering
    \includegraphics[width=1.8\columnwidth]{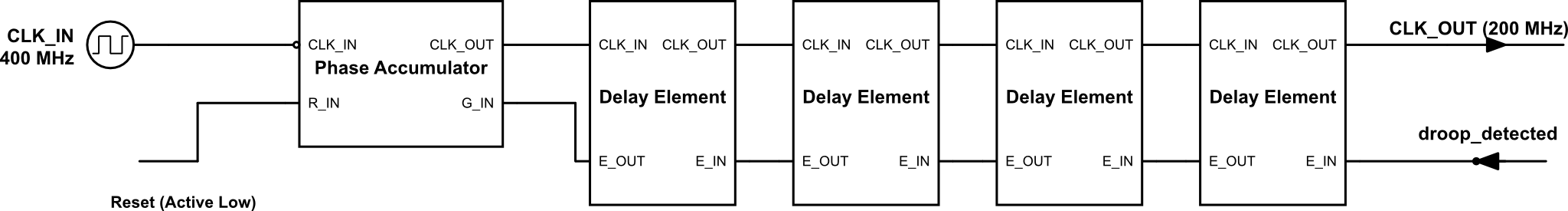}
    \caption{Frequency Adaptation Module: The high-level schematics are identical to~\cite{DroopJournal}, with the exception that the input frequency is twice the output frequency. This allows to add phase delay to the output clock of default period $T=50$\,ns in increments of $T/4$ without needing accurate delays within the phase accumulator implementation. The phase accumulator decides in each clock cycle whether to delay the next rising clock transition by $T/4$ based on the binary input \signal{G\_IN}. The delay elements form a synchronizer chain ensuring that the \signal{droop\_detected} signal, which is sampled with each element's output clock, exposes the phase accumulator to a negligible probability of setup/hold time violations only. However, the delay elements simultaneously serve the purpose of immediately adding a phase shift of $T/4$ whenever the \signal{droop\_detected} signal is stable $1$ upon being sampled. This requires a careful design of the delay element that guarantees a glitch-free clock signal with feasible timing even when the delay element suffers from metastability due to sampling a non-stable value at \signal{E\_IN}.}
    \label{fig:high_level_FAM}
\end{figure*}
\subsection{Our Contributions}
We introduce a latch based design of this circuit which does not use PLLs. Further, we specify and implement \emph{masking latches,} which replace the flip flops in~\cite{DroopJournal}.\footnote{Alternatively, one might take the view that we implement masking registers in which a pair of registers shares the first latch.} This surfaces challenges for the timing properties of the circuit that were hidden by the abstraction of masking registers. Further, for the purposes of completeness, we provide a rudimentary droop detector design, which we used in our chip. We elaborate on aspects of the design motivated by practical implementation issues, such as the asymmetry of buffer delays, the debilitating issues that arise when a proper PLL is unavailable as a cell to stabilise delay lines against process, temperature, and voltage variations, and the avoidance of a glitch in the clock signal that arises from violating certain timing constraints. We successfully implemented our design in the IHP 130\,nm technology.

\subsection{Paper Organisation}
In \Cref{sec:design}, we introduce our design of our circuit. This section includes a subsection for each component that describes their functionality as well as sketches of the proofs of correctness. In particular, we introduce our masking latches in this section (see \Cref{subsec:masking}). In \Cref{sec:synthesis}, we briefly describe the design choices for the RTL synthesis and layout of our test chip, including parameters like clock frequency and choice of IO ports. \TODO{Rewrite this properly at the end of the writing process}

\section{The Design}\label{sec:design}

In this section we elaborate on the design changes in the components of our circuit as compared to \cite{DroopJournal}. At a high level, the schematic is identical, cf.~\Cref{fig:high_level_FAM}, but some crucial design changes need to be made to the constituent components to ensure correct timing properties for a latch based design without accurate delay lines. First, we discuss our implementation of a masking latch, which is the key element enabling our design to mask metastability, guaranteeing a glitch-free output clock signal.

\subsection{Masking Latches}\label{subsec:masking}
\begin{figure}[t]
    \centering
    \includegraphics[width=0.45\textwidth]{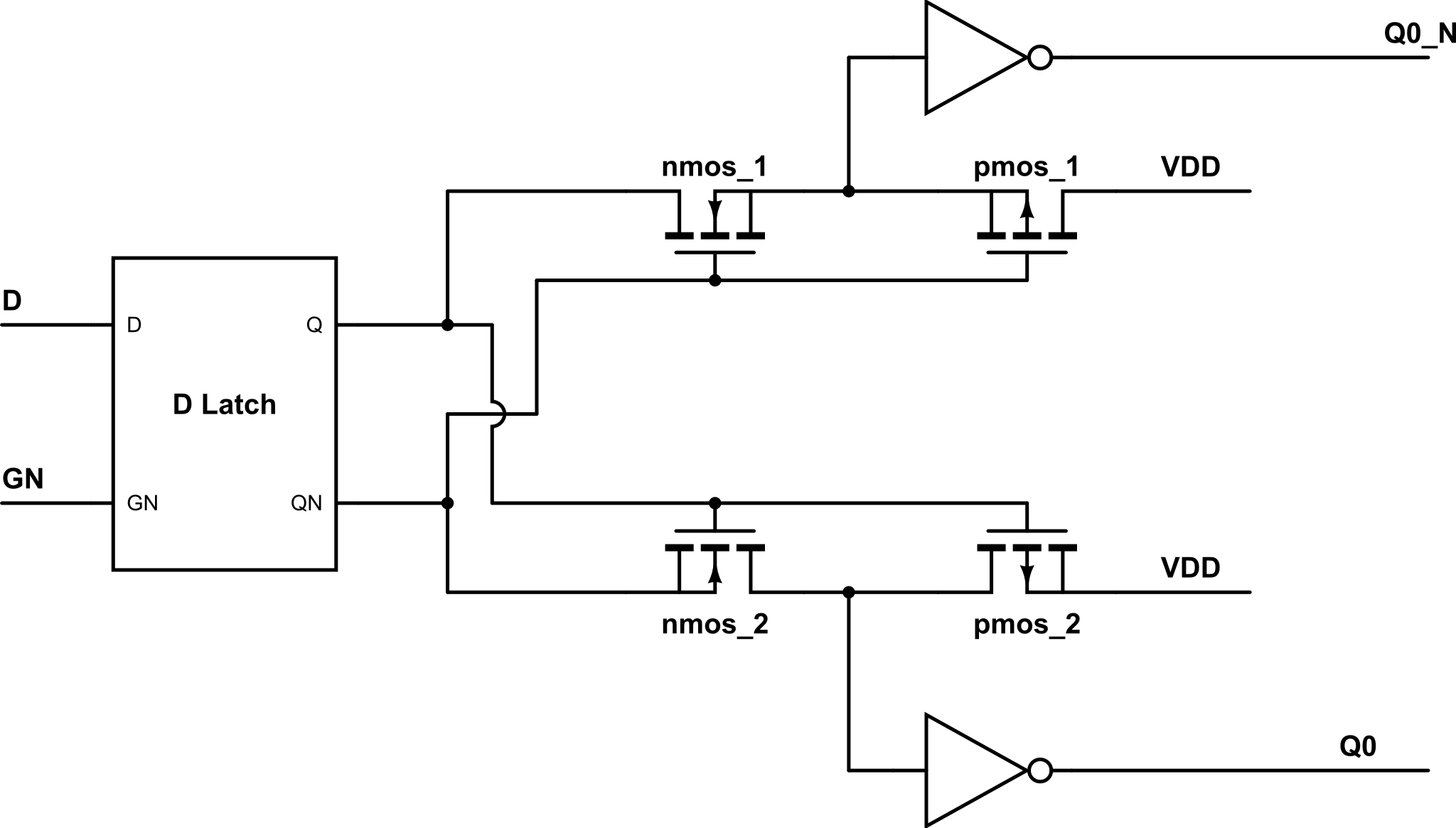}
    \caption{Masking latch. The latch provides a $0$-masking output, i.e., it implements an \signal{Mask-0\_Latch}}
    \label{fig:MC_Latch_0}
\end{figure}

\begin{figure}[t]
    \centering
    \includegraphics[width=0.45\textwidth]{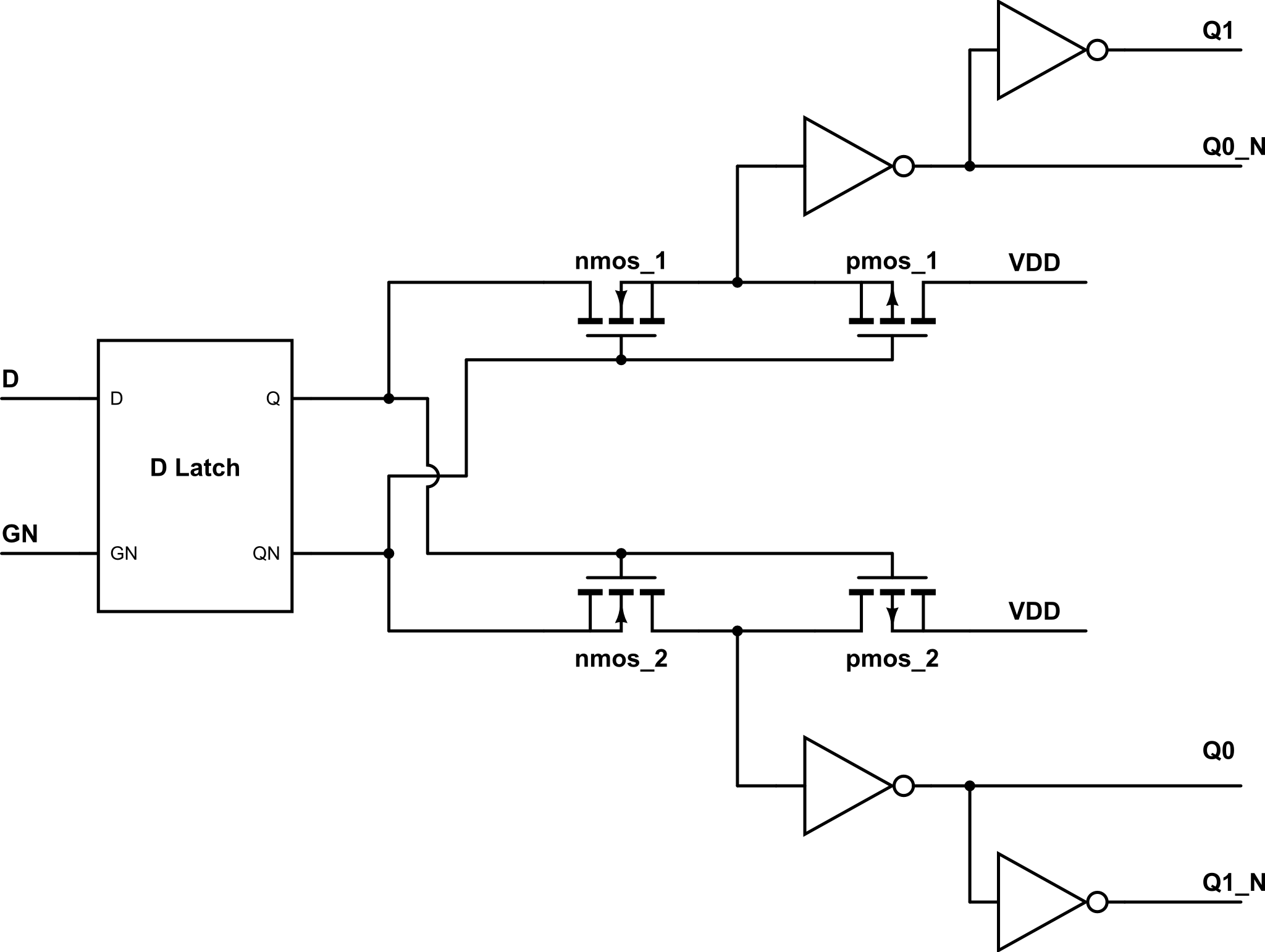}
    \caption{Masking latch that provides both a $0$- and a $1$-masking output, i.e., it implements both \signal{Mask-0 Latch} and \signal{Mask-1\_Latch}.}
    \label{fig:MC_Latch_01}
\end{figure}

\begin{comment}
    \begin{figure}[t]
    \caption{Transient diagrams demonstrating the behaviour of \signal{Mask-0\_Latch} with metastability in the internal latch.\TODO{\textbf{To Ian} could we use your SPICE simulations image?}}
    \label{fig:MC_Latch_transient}
\end{figure}    
\end{comment}

A key component are the masking latches, which in our design perform four functions: 
they sample and store data, synchronize the sampled data signal
to the clock, and prevent any metastable voltages arising from 
the signal sampling to be delivered to their outputs, but instead present valid logic levels at the latch outputs.
The complete synchronization function takes place over cascaded
masking latch stages, each latch providing additional orders of
magnitude of reliability.
Note that although metastable voltages are prevented from 
propagating to these latch outputs, internal metastability
in the latches will still delay when the latch output changes 
from the masked logic level to the sampled logic level.
Thus, metastability can still be propagated
in the form of non-deterministic output transition delays.

Masking latches come in three forms: mask-0, where the output is 
held at logic-0 when the latch internally is at a metastable
voltage; mask-1, where the output is held at logic-1 when the 
latch is metastable; and mask-01, which has two outputs,
one of which is at logic-0 and the other at logic-1 when the 
latch is internally metastable.
Under the most common situation, when sampling the input data 
is unambiguous and any metastability is non-existent or extremely
short-lived, all three of these forms of latches operate 
deterministically and behave as standard latches with a 
predictable clock-to-q output delay. 
Only in rare cases will metastability occur resulting in 
prolonged clock-to-q output delays.

One way of implementing the 
masking function is to use shifted threshold inverters at the output, see e.g.~\cite{982426}. 
Since the internal storage loop of the latch is metastable at a different voltage,\footnote{Oscillatory metastability can be avoided by appropriate design, so that metastability results in a stable intermediate voltage until it resolves.} the output transition occurs only once the latch is recovering from internal
metastability.
This solution suffers from the drawback of the threshold voltage being sensitive to supply voltage variations 
as well as to transistor fabrication variations in both p and n-type transistors, whose characteristics are defined in different fabrication processing step.
This is in conflict with the goal of minimizing energy consumption and thus supply voltage, which created the need for an adaptive response to droops in the first place. 

Instead, we propose to use a differential output sensing circuit, see \Cref{fig:MC_Latch_01}, which produces an output transition only when the two nodes of the bistable circuit element in the latch have separated 
by a large enough voltage indicating that it is exiting from the metastable stage and resolving to a distinct logic level. Such a circuit can take 
advantage of using transistor gate-to-source threshold voltages, which are stable over a wide range of supply voltages.
Moreover, it relies on only one type of sensing transistor, limiting the relevant process variations to a more tightly controlled fabrication process step.

\subsubsection*{Circuit Design Approach}

Adding additional circuit elements to the nodes of the bistable 
circuit element in the latch exponentially degrades its metastability 
resolving capability. As explained in \cite{cox2015}, the transistor sizes in synchronizer and data latches are tuned very differently. 
Data latches are tuned to have a short clock-to-q delay, 
while synchronizer latches need to be tuned to speed up their 
ability to resolve from a metastable state to a defined logic value, minimizing the characteristic resolving time constant $\tau$. 
% Latches with a small $\tau$ value resolve much faster.
Small tau values can be achieved by making the bistable circuit element of the latch several times larger than in a typical data latch.
By starting with a latch tuned to be a synchronizer, the addition of 
either the differential output sensing circuit or the shifted threshold
inverters can replace the output inverter that the latch would otherwise have. 
This results in only a slight increase in $\tau$, 
and thus has only a minor impact on the delay caused by resolving metastability.
Metastability analysis simulations, using the technique described in \cite{molnar1992,yang2007a,yang2007b}, of a well-designed synchronizer latch and mask-01 version of this latch using the differential masking circuit shows a marginal increase of $\tau$, from 106 to 108\,ps, see \Cref{figure:Standard_latch_metastability_analysis_plots,figure:Mask_01_latch_metastability_analysis_plots}. 
Such values of $\tau$ are quite typical for good synchronizer latches in this generation of process technology\TODO{Do we need a citation for this?}.

\begin{comment}
\TODO{Better as (shared) caption: Each of these metastability plot results present a pair of graphs.
On the left is an overlay of multiple simulations presenting Voltage vs Time. 
As the latch is driven further into metastability its output
exhibits additional delay, in this case of several hundred pico-seconds.
The corresponding resolution time semi-log plot of these addition delay results is shown to the right, and from a linear fit to the data values the slope of the fitted line gives the resolution time constant value, $\tau$, for the circuit.}

\end{comment}
Observe in \Cref{figure:Standard_latch_metastability_analysis_plots} that the synchronizer latch propagates metastable voltages at its output that lie between the logic-0 and logic-1 values, 
and when metastability resolves the output rapidly settles to a solid 
logic-0 and logic-1 value.
In contrast, in \Cref{figure:Mask_01_latch_metastability_analysis_plots} the mask-0 latch masks the internal mid-voltage values and outputs a constant logic-0 until any metastability
has resolved, as intended.

\TODO{Ian: Paragraph on MTBF when using the 4 delay element stages in our design. \textbf{Shreyas:} This doesn't serve any purpose since we can't get these stats in our experiments anyway. \textbf{Christoph:} The purpose is to convey in numbers that it's reliable, rather than letting the reader having to derive this on their own. We cannot assume that they are very familiar with this topic. Using simulation results in lieu of experimental data makes perfect sense if we can't have the latter.}

\begin{figure*}[t]
\includegraphics[width=0.95\textwidth]{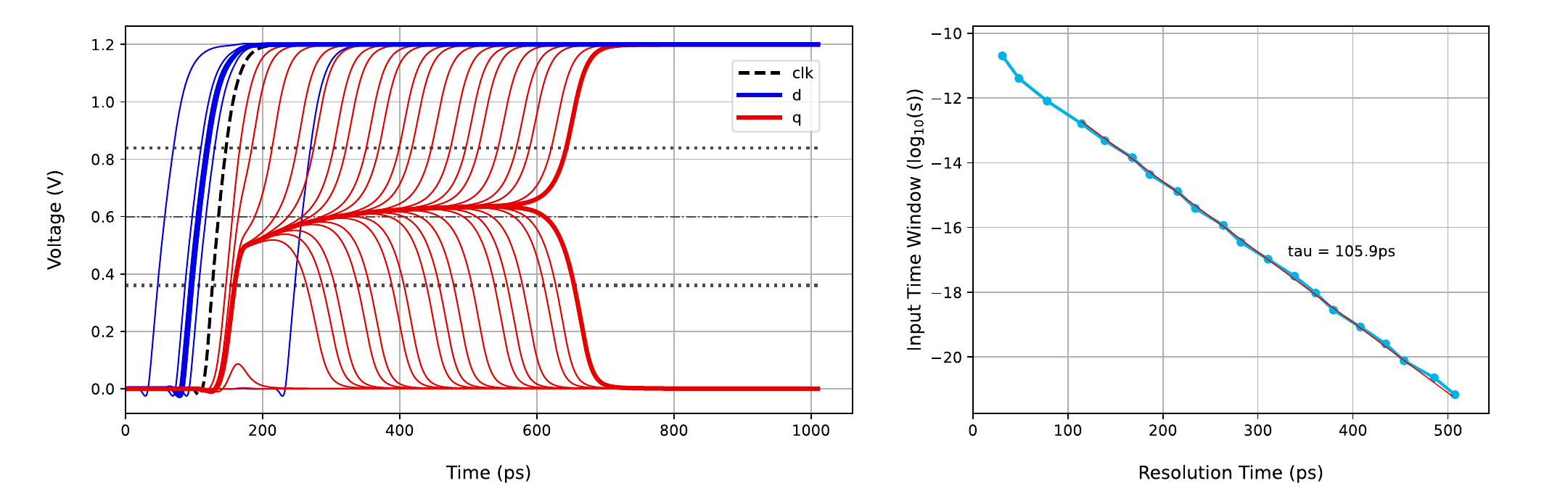}
    \caption{Synchronizer latch metastability analysis result plots.\label{figure:Standard_latch_metastability_analysis_plots}}
\end{figure*}

\begin{figure*}[t]
\includegraphics[width=0.95\textwidth]{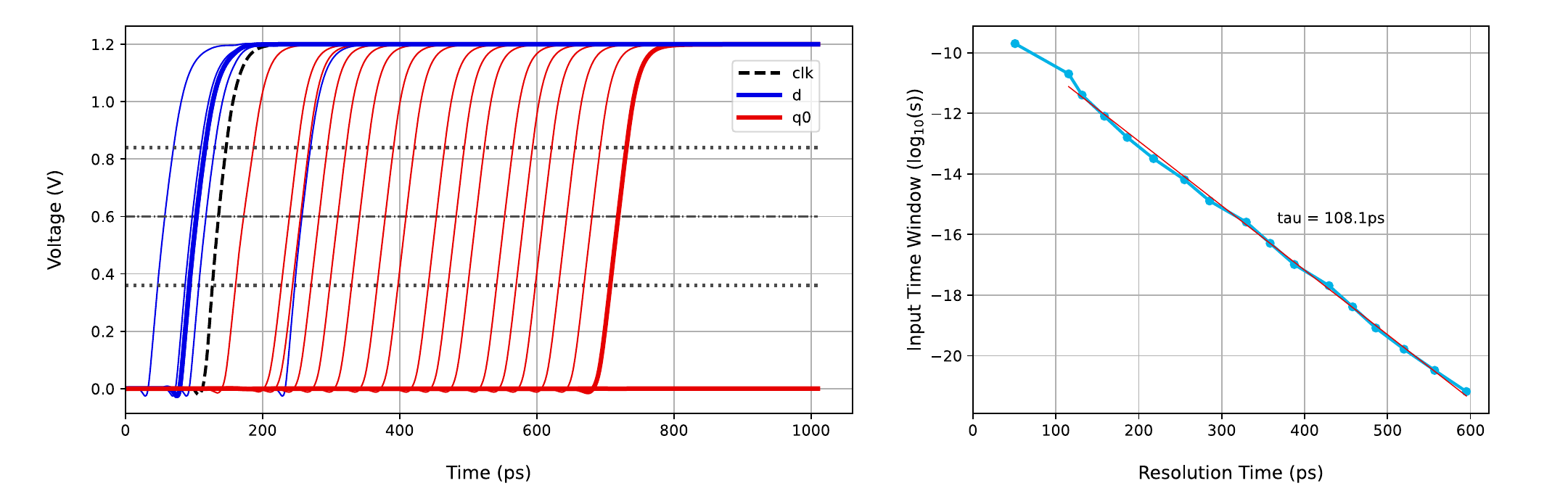}
    \caption{Mask-0 latch metastability analysis result plots.\label{figure:Mask_01_latch_metastability_analysis_plots}}
    \caption*{Each of these metastability plot results present a pair of graphs.
On the left is an overlay of multiple simulations presenting Voltage vs Time. 
As the latch is driven further into metastability its output
exhibits additional delay, in this case of several hundred pico-seconds.
The corresponding resolution time semi-log plot of these addition delay results is shown to the right, and from a linear fit to the data values the slope of the fitted line gives the resolution time constant value, $\tau$, for the circuit}
\end{figure*}

\subsubsection*{Containing Metastability}
In order to clarify the meaning of metastability containment, we provide a brief explanation of the challenge presented by metastability.
Recall that any bistable element, in addition to its two digital states, can enter a metastable state. 
When this happens in a standard latch, the analog output signal of the latch is stuck at a voltage which does not fall into either the \signal{logic 1} window or the \signal{logic 0} window in which its output can be safely treated as a digital 1 or 0 bit, respectively.
Metastability usually occurs in a latch or a flip flop when there is ambiguity in the voltage of the data input signal at the time that it is sampled. This can occur when the data input is an analog signal of arbitrary input value or when the data signal is transitioning at the time of sampling, i.e., setup time and hold time violations. Marino~\cite{marino1981general} examined this phenomenon mathematically and using the theory of continuous functions on connected topological spaces, showed that the possibility of metastability is unavoidable in bistable elements, i.e., no circuit can guarantee to remove metastability deterministically.
Further, the metastable output might settle arbitrarily to zero or one---potentially at a most inconvenient point in time. However, the result of Marino does not rule out the possibility of a device that has the following features:
\begin{itemize}
    \item In the absence of metastability inducing conditions, it behaves like an ordinary latch or flip-flop.
    \item If a setup/hold time violation or ambiguity in an analog input voltage occurs when the device is latched, there is at most one (possibly incomplete) output transition after which the output settles into a stable logic value, prior to the next time the device is latched.
    \item It can be specified whether the possible output during internal metastability of the latch is \signal{logic 0} or \signal{logic 1}, i.e., whether the latch or flip-flop masks internal metastability to $0$ or to $1$, respectively.
\end{itemize}
We emphasize that this behavior assumes a clean (i.e., glitch-free) digital signal at the clock input. Even so, the output transition that might follow a setup/hold time violation can occur arbitrarily late, just as there is no deterministic bound on when the output of a synchronizer stabilizes. Moreover, even if the data signal is glitch-free, it is not guaranteed that the output signal is glitch-free. After a setup/hold time violation, it is possible that a late output transition occurs just when the clock input is about to transition as well. This could result in a too quick transition in the other direction or ``choke off'' an output transition that has just begun; this is what the ``possibly incomplete'' in the above specification refers to.

\subsubsection*{Circuit Specification and Implementation}
The authors of~\cite{DroopJournal} proposed an abstract masking latch based design, arguing its feasibility based on the shifted threshold approach given in~\cite{982426}.
In contrast, we fully specify and implement superior masking latches based on differential sensing.
To this end, we adapt standard latch cells and make them metastability containing.

Our metastability-containing latches implement the above behavior by augmenting a traditional latch with a differential sensor. This design was originally proposed by Chuck Sheitz in the textbook of Mean and Lynn\cite[Chapter 7]{meadlynntextbook1980} in the context of designing Q-Flops for self-clocked circuits. In such circuits, a Q-Flop is part of the arbiter of the self-clocking machinery that decides when the next clock pulse must be generated. These generated clock pulses are in turn used to clock the Q-Flop. This helps them avoid data races by picking when they sample next. The metastability containing behavior there prolongs an existing clock signal until metastability is resolved and delays the generation of the next clock pulse. Because their clocks are internally generated, they can reason about their circuit entirely with no reference to device delays and in a self-contained manner.

Note that we use the circuit in a fundamentally different manner, as we use our latches in a setting where they have no control over the signal that clocks them. This places stronger design constraints on their usage and one needs to reason about their correctness conditioned on the properties of the input clock signal of the modules they are part of. Such reasoning will also have to consider device delays. We design two variants of the latch, one that defaults to \signal{logic 0} and one that defaults to \signal{logic 1} when metastability occurs, and use them correctly to produce good output signals.

\Cref{fig:MC_Latch_0} shows the schematics of our latch and \Cref{figure:Mask_01_latch_metastability_analysis_plots} its behavior under time-varying inputs. We first consider the specification of an \signal{M0\_latch}. Such a latch has ports similar to an ordinary D-latch, namely data input \signal{D}, negative enable input \signal{GN} and latch outputs \signal{Q0} and \signal{Q0\_N}. Internally, our latch uses a standard \signal{D-Latch}.
\begin{itemize}
    \item When the internal latch produces clearly digital output values at its output ports \signal{Q} and \signal{QN}, the output ports \signal{Q0} and \signal{Q0\_N} must produce voltages that are unambiguously in the ranges of the respective logic levels. Further the latch must satisfy the following equations for the logic levels of the output : $\signal{Q0} = \signal{Q}$ and $\signal{Q0\_N} = \signal{QN}$.
    \item When the internal D latch produces a metastable voltage level on its output ports \signal {Q} and \signal{QN}, both output ports \signal{Q0} and \signal{Q0\_N} must produce an output voltage that is unambiguously logic 0.
\end{itemize}
A further consideration is that this circuit must respond fast in order to provide meaningful metastability masking. Once we have an \signal{M0\_Latch}, we can obtain an \signal{M1\_Latch} by careful inversion of the output signals and selection of ports.

We now discuss the operation of our \signal{M0\_Latch} (see \Cref{fig:MC_Latch_0}) in the absence and presence of metastability. For the sake of notational convenience, we label the output signals of the internal latch as \signal{Q\_int} and \signal{QN\_int}. 
First, assume that all signals are settled to stable values.
In this case, either \signal{Q\_int\ =\ 1} and \signal{QN\_int\ =\ 0} or \signal{Q\_int\ =\ 0} and \signal{QN\_int\ =\ 1}. In this case, the differential amplifier can be viewed as being composed of two inverters numbered 1 and 2, whose ground inputs are respectively \signal{Q\_int} and \signal{QN\_int} and whose gate inputs are \signal{Q\_int} and \signal {QN\_int} respectively. Without loss of generality, we assume \signal{Q\_int = 1} and \signal{QN\_int = 0}. For inverter 1, the ground signal is equal to VDD, but \signal{nmos\_1} is turned off, while \signal{pmos\_1} is turned on. Thanks to the inverter, placed to give a sharp transition, \signal{Q0\_N = 0}. At the same time, for inverter 2, the ground signal is 0. Thus inverter 2 functions as an ordinary inverter gate with input $1$. Combined with the inverter at the output, we get \signal{Q0 = 1}.

This static analysis is inadequate to reason about the behavior of the differential sensor when the internal latch becomes metastable. Instead, we must consider the analog behavior of the differential sensor. For this discussion, we denote the voltages on \signal{Q\_int} and \signal{QN\_int} by $V_Q$ and $V_{QN}$ respectively. When the internal latch is metastable, we have that $V_Q \sim V_{QN} \sim $ VDD$/2$. The gate and source voltages of both nmos transistors \signal{nmos\_1} and \signal{nmos\_2} are identical, making them completely opaque. On the other hand, there is a significant voltage gap between the source and gate voltages of the pmos transistors, which provides pull up on the drain terminals of both \signal{pmos\_1} and \signal{pmos\_2}. Thanks to the inverters, both \signal{Q0} and \signal{Q0\_N} thus output 0.

This is to be contrasted with the situation in which the internal latch is in a stable state:
for instance, when \signal{Q} is at logic 1 and \signal{QN} at logic 0, then nothing changes in the behavior of \signal{nmos\_1} and \signal{pmos\_1}, but \signal{nmos\_1} is opaque while \signal{pmos\_1} has a stronger pull-up effect on its drain terminal due to the larger difference between its source and gate voltage. However, \signal{pmos\_2} is now opaque while \signal{nmos\_2} is transparent. Hence, there is a strong pull-down at the drain of \signal{nmos\_2}, which causes \signal{Q0} to be at a voltage equivalent to logic 1. Thus, digital latch behavior is achieved whenever one of the \signal{nmos} latches is conducting and exerts a strong pull down effect on its drain. 

To meet the fast response time requirement under metastability, it is important to consider the capacitance load on the outputs of the internal latch as well as at the output taps of the differential sensor. The effect can be exacerbated by the smaller differential in the smaller pull-up effect at the \signal{pmos transistors} when the internal latch is metastable. 
\begin{itemize}
    \item If the capacitive load produced by the transistors is large, the response of the pmos transistors to the gate-source voltage difference is exponentially slower.
    \item  Likewise, a large inverter at the output of the differential sensor also slows down the response time. 
\end{itemize}
Thus, it is imperative to use the smallest possible transistors for the differential sensor and choose the smallest inverters possible for their outputs.

\subsection{Phase Accumulator}
\begin{figure*}[t!]
    \centering
    \includegraphics[width=0.7\textwidth]{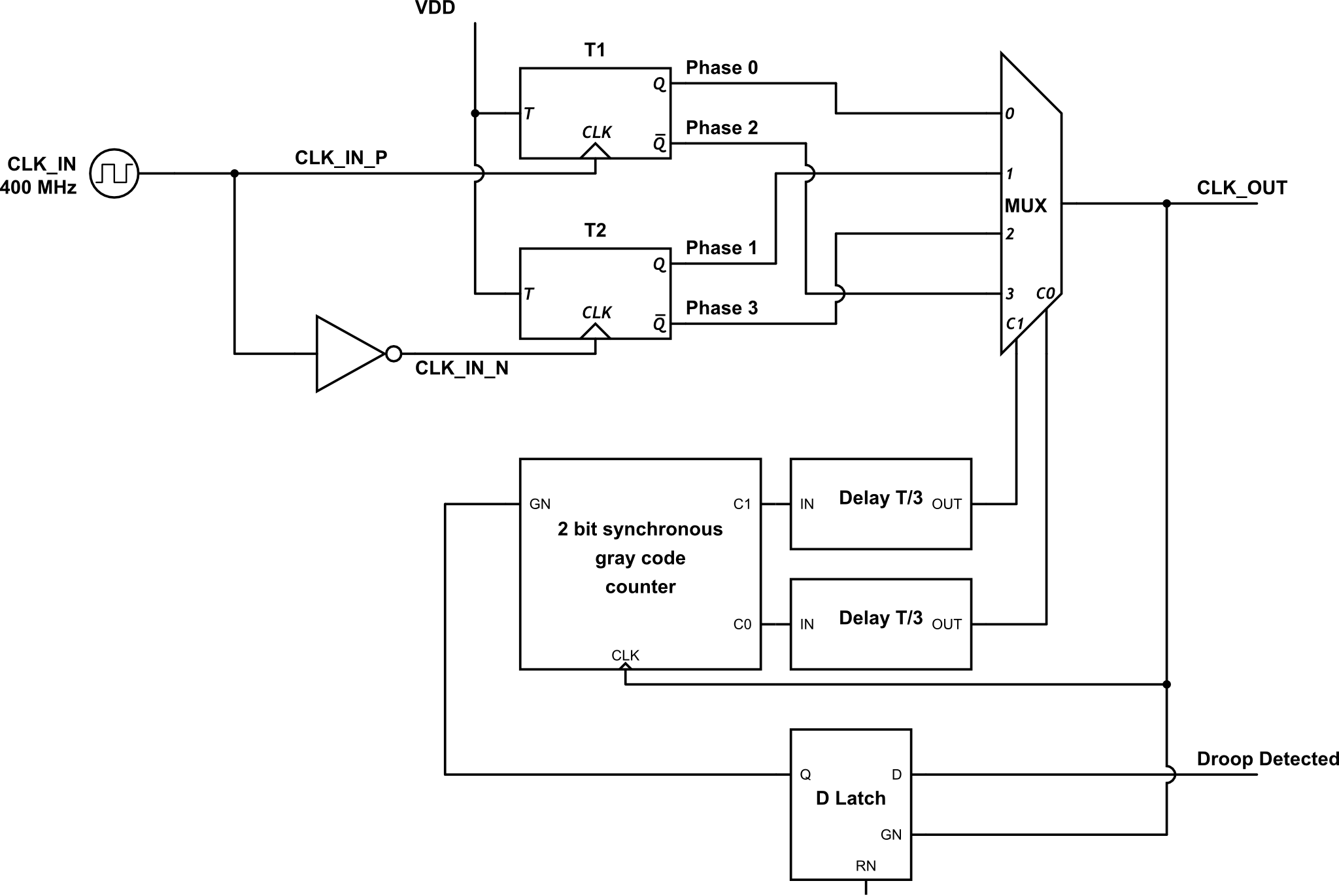}
    \caption{The Phase Accumulator: Generating Four Clock Phases\label{fig:Phase_Accumulator_four_phases}}
\end{figure*}   

In the absence of a droop, the phase accumulator simply divides down the input clock by a factor of $2$. However, when it samples \signal{G\_IN} to be high, it ``drops'' half a cycle of the input clock, i.e., leaves \signal{CLK\_OUT} low for additional $T/4$ time.

The main challenge here is that $T$ is gauged to match the standard clock speed for the used technology, leaving too little time even for extremely simple calculations based on the input clock. We resolve this by first deriving $4$ clocks of period $T$ and relative phase shifts of $0$, $T/4$, $T/2$, and $3T/4$ from the input clock. A multiplexer is used to select the output matching the current phase shift, and the respective state is held by a modulo $4$ counter. We exploit that the decision to switch to a different input of the multiplexer, i.e., whether the counter was increased or not, becomes only relevant $T$ time after the rising clock flank of the currently used of the four divided clocks. This leaves sufficient time to increase the counter and stabilize its output in time. What remains is to take care that no output glitches arise from switching to the next input of the multiplexer. This is achieved by a suitable counter design, pre-computing the counter output, and properly timing the latching of the flip-flop capturing this output.\footnote{Technically, this also requires that the multiplexer is glitch-free when switching between different input signals with the same value. This is satisfied by standard designs, however, which we verified in simulation for the relevant technology.}

Note that the phase accumulator would already solve the task at hand in the absence of setup/hold time violations of its \signal{G\_IN} input. The purpose of the chain of delay elements is to guarantee this up to a negligible probability of failure, without increasing the response time of the circuit by the time used for synchronization.

A key difference to the design from~\cite{DroopJournal} is that we assume an input clock which is twice as fast, so that we can use T-Flip Flops instead of internal delays to produce clocks with phase offsets. This carries the advantage of not requiring a PLL to stabilize these delays against process and other variations. In the design from~\cite{DroopJournal}, failing to do so would result in accumulating these variations over all phase shifts, i.e., up to a total delay of $3T/4$, before switching back to a phase offset of $0$. This final step, corresponding to a phase shift of $T/4$, then would be subject to the full accumulated error, including the possibility of a drastically shortened high time of the clock output signal.

A few details on our implementation, shown in~\Cref{fig:Phase_Accumulator_four_phases}:
\begin{itemize}
  \item We use T Flip Flops labelled \signal{T1} and \signal{T2}. The input signal \signal{CLK\_IN} is used to produce two copies, which we label \signal{CLK\_IN\_P} and \signal{CLK\_IN\_N}, such that $\signal{CLK\_IN\_P} = \signal{CLK\_IN}$ and $\signal{CLK\_IN\_P} = \signal{NOT}\  \signal{CLK\_IN}$, respectively, are fed into the \emph{clock} inputs of \signal{T1} and \signal{T2}. Their \emph{toggle} inputs are tied to \signal{VDD}. Thus, the outputs \signal{Q} of \signal{T1} and \signal{T2} alternate every time there is a rising clock edge on \signal{CLK\_IN\_P} and \signal{CLK\_IN\_N}, respectively. This yields clocks which are at phase offsets $0$ and $\pi/2$. The other two phase offsets are obtained by the negation of these two clocks. This solution has the benefit of reliably producing evenly spaced phase offsets. Additionally, there is no need to redesign this part of the circuit for a different clock period $T$, unlike one based on buffer cells, whose delays are static w.r.t.\ the clock signals.
  \item The second consideration is about the time at which the effect of activation of the \signal{droop\_detected} signal causes the multiplexer to select a different clock with a phase shift. In particular, we wish to ensure the following:
  \begin{itemize}
    \item The \signal{Droop Detected} signal is received from the delay element to the right, which is clocked by the \signal{CLK\_OUT} output of the phase accumulator. To correctly time the latching of the next data value from the synchronizer pipeline the chain of delay elements implements, both the latch and the counter are clocked by the \signal{CLK\_OUT} signal as well.
    \item Recall that an up-count of the counter changes the control input of the multiplexer such that we switch to the multiplexer input that has its rising and falling flanks roughly $T/4$ time later than the preceding one. Thus, applying the change immediately would result in the multiplexer output transitioning to low and then high again, i.e., cause a glitch in the output clock. We prevent this by delaying the output signal of the counter by roughly $T/3$; this is larger than $T/4$, so no glitch is introduced due to switching too early, and smaller than $T/2$, so no glitch is introduced by failing to switch before the next clock flank of the ``old'' select input used. Choosing a delay value bounded away from both $T/4$ and $T/2$ guards against PVT variations.
    \item We require that the high time of the signal is maintained even when delaying the next rising clock edge at the output by $T/4$. Accordingly, we seek to switch to the multiplexer input with $T/4$ larger phase offset after the falling clock edge. We do so by using a counter which samples its input and performs an up-count on the falling edge of its clock input. This also leaves sufficient time for the counter to take note of the value the latch captured\footnote{The latch becomes transparent again on the falling edge, but it is easy to ensure that the hold time of the counter is smaller than the propagation delay through the latch.} on the rising edge of \signal{CLK\_OUT}.
    \item Lastly, to avoid near-simultaneous transitions at the selection inputs of the multiplexer, which can also lead to a glitch, we use a Gray code counter and ensure that in each step, only one bit of the counter's output changes.
  \end{itemize}
  \item Our latches and droop detected signals are both active low. Further, the reset output on all latches is also low. This creates a situation where the output of the latch that latches \signal{droop\_detected} is identical when a droop is asserted and when the latch is reset. Thus, on reset the data inputs to the latches match the reset value and no metastability is induced by an ill-timed de-assertion of the reset signal. Therefore, we latch the reset signal with \signal{CLK\_OUT} before passing it into the phase accumulator. Note that this implies that the first rising clock flank traveling through the chain of delay elements will be delayed at every element. However, the phase accumulator will apply matching phase shifts to all subsequent clock flanks, so this is of no consequence to the circuit's operation.
\end{itemize}
\begin{theorem}\label{thm:pa}
Suppose that dynamic changes in delays between consecutive clock cycles are negligible, an inverter delay is sufficiently small, and that delay variations are sufficiently small. Moreover, assume that on an up-count, the counter only changes the output bit that transitions, and the multiplexer has no output glitch when the control signal indicates a change from one input that is high to another that is high. Then the circuit given in \Cref{fig:Phase_Accumulator_four_phases} correctly implements a phase accumulator in the sense of~\cite{DroopJournal}.
\end{theorem}
\begin{proof}[Proof Sketch]
Granted that the input clock signal has high and low time $T/2$ and there is no setup/hold time violation on the \signal{Droop Detected} input, we need to show that the output signal is glitch-free with high time $(1\pm \varepsilon)T/2$ and low time $(1\pm \varepsilon)(T/2+b\cdot T/4)$ on each cycle, where $b$ is the bit sampled by the D-latch in the respective cycle.

First, as discussed above, the multiplexer receives derived clocks of frequency $200$\,MHz with phase shifts of roughly $0$, $T/4$, $T/2$, and $3T/4$ relative to the input clock. We now show the claim by induction on the cycle number. As explained above, the reset to low asserts the same value as the \signal{Droop Detected} signal provides on boot-up, so regardless of the precise timing of the de-assertion of the reset signal, the D-latch will hold a stable $0$ until the D-latch becomes transparent again on the first falling (output) clock flank.

Now assume that for $i\ge 1$ we have shown correct behavior up to the $i$-th rising output clock flank. In particular, on the $i$-th rising output clock flank the D-latch sampled either a stable $0$ or $1$. If a $1$ was sampled, the counter output does not change when it is latched on the falling clock flank. Hence, the select input to the multiplexer does not change, and the $i$-th falling and $(i+1)$-th rising clock flank arriving at the currently selected multiplexer input travel through unimpeded. At this time, the D-latch captures the next data value on the \signal{Droop Detected} input, which the phase accumulator specification assumes to be stable by this point.\footnote{To compose the phase accumulator with a chain of delay elements, one needs to check the respective timing as well. However, delay elements clock their latches with the input clock and maintain the value they latched in the preceding cycle until the falling flank of the current cycle. Accordingly, the timing constraints resulting from connecting a delay element's \signal{CLK\_IN} and \signal{E\_OUT} to the phase accumulator's \signal{CLK\_OUT} and \signal{Droop Detected}, respectively, are virtually impossible to violate.} 

Now consider the case that a $0$ was sampled on the $i$-th rising output clock flank. The counter then performs an up-count on the $i$-th falling output clock flank, which affects the control input of the multiplexer more than $T/4$ time later, but less than $T/2$ after the $i$-th rising clock flank at the output (here we use that delay variations are small enough). Thus, the change occurs when the previously selected multiplexer input is low. As the one that is selected now has (up to an inverter delay and delay variations) a phase offset of additive $T/4$, it is also low at this point in time. Note that this also applies when the counter overflows from $3$ to $0$, as the period of each of the four clock signals fed into the multiplexer is $T$. Because we use a Gray code counter, only one output bit is affected by the up-count. By assumption, this means that only one of the control inputs to the multiplexer is affected. Again by assumption, this implies that the multiplexer output does not glitch on the transition. Finally, since the transition is made after the $i$-th falling flank of the previously selected multiplexer input, the $i$-th high time of the signal is (up to delay variations) $T/2$, and the subsequent low time is increased by roughly $T/4$.
\end{proof}

\subsection{Delay Element}
\begin{figure*}[t!]
        \centering
        \includegraphics[width=0.95\textwidth]{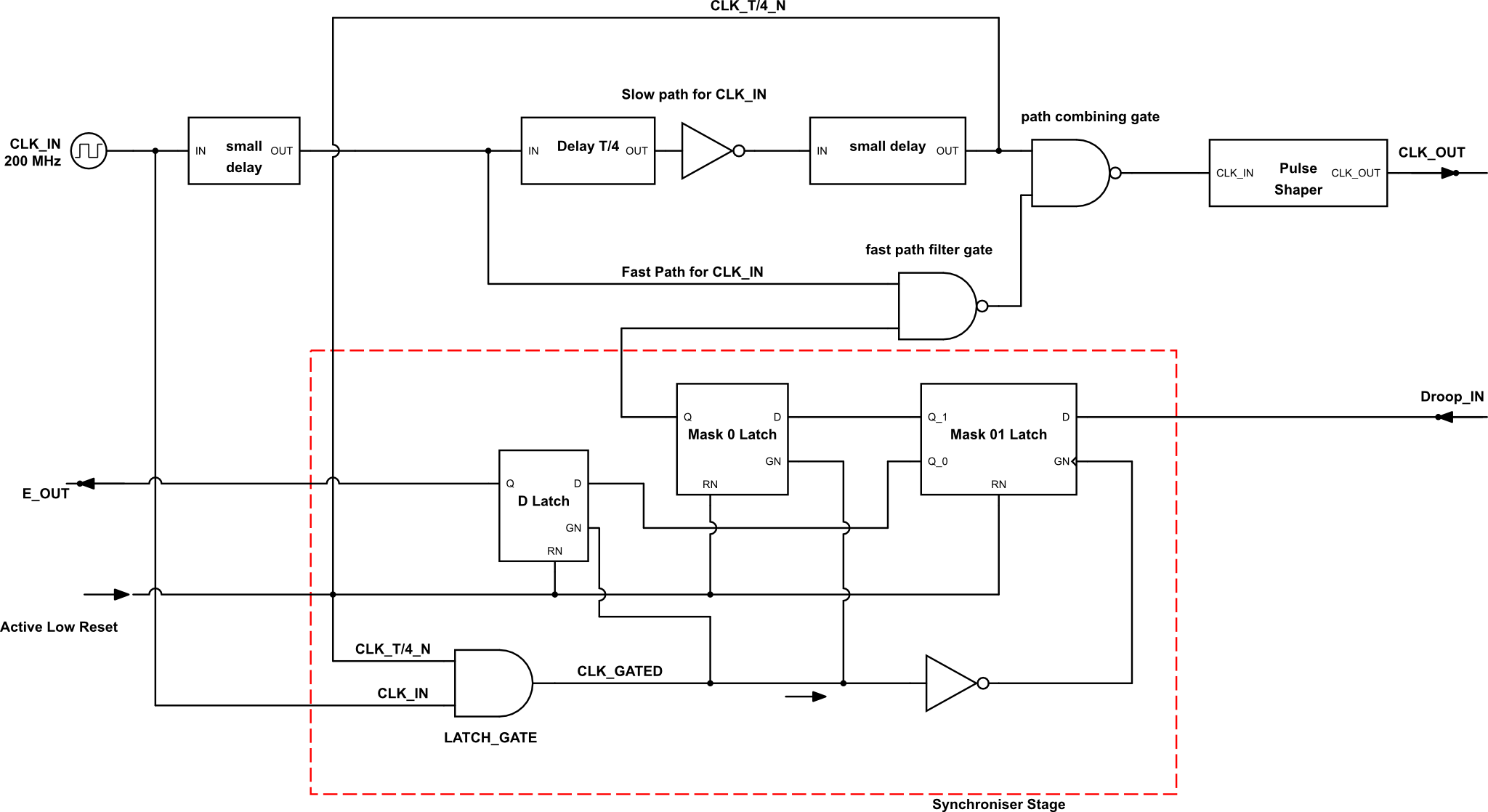}
        \caption{The schematic of our delay element.\label{fig:Delay_Element_Internals}}
\end{figure*}

The delay element is the main innovation of~\cite{DroopJournal}. We need to provide an implementation that uses our masking latch design for the IHP $130$\,nm technology we presented. Further, in order to work without PLLs, our implementation needs to be resilient to delay variations in the constituent cells.

Our implementation is shown in~\Cref{fig:Delay_Element_Internals}. The lower part of the module, from the $01$-masking latch on and below (the part encircled by the red box), implements a synchronizer stage. It receives the \signal{Droop\_IN} signal from either the droop detector or the delay element to the right. The $01$-masking latch acts as the first-stage latch of a synchronizer flip flop and the D-latch as the second-stage latch. There are two key differences in usage to a standard synchronizer chain:
\begin{itemize}
  \item We use the data values propagating through the synchronizer chain to control whether the delay element adds delay or not. This is done via another second-stage latch, \signal{MC 0 latch}, providing one input to a \signal{NAND} gate labelled the \signal{fast path filter gate}. To ensure correctness in face of metastability due to a setup/hold time violation by the \signal{Droop\_IN} signal, this extra second-stage latch and the first-stage latch must have certain masking properties. We stress that we are not implementing masking flip-flops here, but combine masking latches in a very specific way tailored to our application.
  \item In order to align the timing of latching data with its usage by the top part of the module, which is responsible for clock propagation, the latches are clocked by the input clock of the module. Because we are chaining such elements together, the propagation time of the clock through the module needs to be taken into account. In essence, this delay is deducted from the time available in each synchronizer stage, necessitating more careful timing. Note that this also means that the chain has to be slightly longer than a standard synchronizer chain to ensure equally low probability of failure.
\end{itemize}

Let us discuss the operation of the delay element in more detail. A cycle begins with the rising clock flank at the clock input \signal{CLK\_IN}. At this time, no signal transitions are propagating through any of the delay lines in the module with the possible exception of the pulse shaper module, which we will discuss separately. Thus, the \signal{CLK\_T/4\_N} input to the \signal{LATCH\_GATE} gate, at the bottom left corner, has settled to \signal{logic 1}, while the bottom input \signal{CLK\_IN} transitions from \signal{logic 0} to \signal{logic 1}. In other words, the input clock flank passes through the \signal{LATCH\_GATE} gate to the signal \signal{CLK\_GATED} and latches the second-stage latches \Latch{} and \MCZeroLatch, while the first-stage latch \MCZeroOneLatch{} becomes transparent. Due to the inversion of its clock input and the propagation delay through the first-stage latch, the second-stages latches are allowed to properly catch the output of the first-stage latch before it becomes transparent, and then turned opaque. The first-stage latch samples at the falling edge of \signal{CLK\_T/4\_N} and remains opaque for about the next $3T/4$ time units.

Concurrently, the rising input clock flank of \CLKIn{} propagates through the \signal{slow path} indicated in the schematics, which is gauged to guarantee that the \signal{path combining gate}, a \signal{NAND} gate, does not observe the corresponding input transition before its other input had time to settle to the value the $0$-masking latch captured. Observe that the \signal{fast path filter NAND} gate lets the clock flank pass if the latch holds a $1$ and blocks its propagation if it holds a $0$; we will discuss the case of metastability later. Note that the \signal{path combining AND} gate receives the inverted and delayed input clock signal \signal{CLK\_T/4\_N} as top input. Given that the high time of the input signal is at least (roughly) $T/2$, this entails that the second-stage latches become transparent only when the second \signal{fast path filter NAND} gate receives an input of $0$ on its top input. This implies that the $0$-masking latch capturing a $1$ results in a fast propagation of the clock flank, with the second \signal{NAND} gate constantly receiving at least one input of $0$ after arrival of the clock flank until the falling input clock flank has passed through the delay on the top path and been inverted. At this point, the first \signal{NAND} gate already had a stable $0$ at its top input for about $T/4$ time, i.e., the falling clock flank travels unimpeded through the slow top path. On the other hand, the $0$-masking latch holding a value of $0$ closes off the bottom, fast path for the clock signal, meaning that it is simply delayed by $T/4$ (plus the small initial delay that always applies). Either way, the low time of the input clock signal of at least (roughly) $T/4$ is sufficiently long for the delay lines to settle to their original state before the next rising flank arrives.

Observe that the fast propagation resulting from latching a $1$ results in the high time of the output of the second \signal{NAND} gate being extended by roughly $T/4$ relative to the input clock signal. This, in turn, reduces the low time accordingly. To be able to chain multiple delay elements together, the pulse shaper reduces the high time of the signal again to roughly $T/4$, provided neither the high nor the low time of the signal it receives as input are too small.

Finally, observe that the clock inputs to all latches have settled to the values they had at the beginning of the cycle again in time for the next rising input clock flank. Hence, it remains to check that data is propagated correctly when chaining multiple delay elements together. To see this, recall that the master latch is latched when the rising clock flank has propagated through and been inverted by the $T/4$ delay, gone through the \signal{AND} gate, and been inverted again, i.e., slightly more than $T/4$ time after the rising clock flank arrived at the input to the module. Marginally earlier, the D-latch serving as the slave on the data path is becoming transparent; between the time the master latched in the previous cycle (plus propagation to its output and the D-latch) and then, \signal{E\_OUT} continuously has the value latched by the master in the previous cycle. Thus, the master latches the correct value if the latching time lies in the interval spanned by the element to the right (if any) latching its master (plus some small propagation time) due to the previous rising clock flank and making its D-latch transparent again due to the current. We certainly do not latch too late, since the difference of about one inverter delay in the relative latching times of master and slave of the same cell is dwarfed by the time for the rising clock flank to travel from one element to the next. We do not latch too early, since the rising clock flank is, in sum, experiencing a delay substantially below $T$, even on the slow propagation path.\footnote{So far, our discussion merely established that the delay of the rising clock flank is at most roughly $T/4$ plus a couple of small delays and the delay through the pulse shaper. The delay through our pulse shaper implementation will be shown to be roughly $T/10$ plus a few gate delays.}

Overall, each delay element adds a delay of $T/4$ when propagating the clock signal if it latches a stable $1$, and does not add additional delay when latching a stable $0$. Since stable values are handed to the left deterministically, the same conditional delay is applied by the elements to the left and, ultimately, the phase accumulator. Thus, if a clock flank is delayed, the same delay is applied to all subsequent clock flanks.

%###
\paragraph*{Handling Metastability}
%###
As the rightmost delay element samples the \signal{Droop\_In} signal from the droop detector, which might not provide a stable high or low output voltage within the setup/hold time window of the master latch, metastability can be induced. Despite masking behavior of this latch, a late output transition can result in metastability of a slave latch, and this could be propagated to the element to the left. For a sufficiently long delay line, metastability is guaranteed to decay before reaching the phase accumulator with all but negligible probability. However, our design needs to prevent glitches in the clock signal and proper timing of its transitions despite internal metastability of some of the latches in the delay elements.

Accordingly, assume that the master latch becomes internally metastable upon latching a value. By the properties established in \Cref{subsec:masking}, for $b\in \{0,1\}$, until it becomes transparent again, its $Q_b$ output will be stable $b$ with the possibility of a late transition to $1-b$. In particular, only one of its two outputs may undergo a late transition, implying that only one of the slave latches can become internally metastable upon latching. Note that the case that neither becomes metastable is not equivalent to the master latching a stable value, as the slaves then disagree on the sampled value. However, this case is subsumed by the other two, namely one of the slaves becoming metastable, so we skip it in the discussion.

First, consider the scenario in which the $0$-masking latch becomes metastable. If it stabilizes to $0$ or after the rising input clock flank propagates through the $T/4$ delay, is inverted, and reaches the second \signal{NAND} gate, the behavior is identical to the latch capturing a stable $0$. Otherwise, its output has a late $0$-$1$-transition, resulting in a delayed output transition of the first \signal{NAND} gate relative to the case where the latch captured a stable $1$. Thus, the second \signal{NAND} gate has its rising output transition at some point in time from the interval spanned by the timing of this transition in the fast and slow propagation case, i.e., latching a stable $1$ or $0$, respectively. In particular, the timing envelope we established for this transition in the metastability-free case still applies, there is no glitch in the propagated clock signal, and the pulse shaper ensures that the delay element's output signal transitions low about $T/2$ time after going high. Moreover, the D-latch captured a stable $0$, which is propagated to the upstream elements in the chain and ultimately the phase accumulator. Thus, a delay of $T/4$ is applied to all subsequent clock pulses. Overall, proper timing is guaranteed, although it may happen that only a ``fractional delay'' of $x\in [0,T/4]$ is applied to the current clock pulse. 

Next, suppose that the D-latch becomes metastable. This entails that the $0$-masking latch captures a stable $1$, i.e., no additional delay is applied by the current element. It is now possible that the master latch of the element to the left becomes metastable upon latching on the next falling clock flank traveling through. Again, we need to consider both cases for this delay element. First, if the $0$-masking latch captures a stable $1$, this is equivalent to fast propagation of the clock signal. By induction from right to left in the chain, this case applies to all delay elements until the one where metastability gets resolved. At this element, where the second case applies, a delay of $x\in [0,T/4]$ is added. Subsequently, all clock pulses are delayed by $T/4$. This is the crucial property causing us to chain together masking latches exactly in the way done here: it must never happen that a clock pulse is delayed, but subsequent pulses are not.

%###
\paragraph*{Showing Correctness}
%###
In~\cite{DroopJournal}, correctness of the overall circuit is proved based on the correctness of the submodules. Accordingly, it is sufficient to show the required properties of the delay element. Under the assumption that the input signal to each delay element is glitch-free, has high time of $(1\pm\varepsilon)T/2$ (for a sufficiently small constant $\varepsilon>0$), has low time of at least $(1-\varepsilon)T/2$, and has consecutive rising flanks separated by $T$ time, these are the following.
\begin{enumerate}
  \item A chain of delay elements acts as a (synchronous) data pipeline, where values are handed to the element on the left on each clock pulse.
  \item The falling clock flank at \signal{CLK\_OUT} follows always $(1\pm \varepsilon)T/2$ time after the rising clock flank.
  \item If a stable value of $1$ is sampled on a rising clock flank, this clock flank is forwarded to \signal{CLK\_OUT} with some fixed delay $\delta$.
  \item If a stable value of $0$ is sampled on a rising clock flank, this clock flank is forwarded to \signal{CLK\_OUT} with delay $\delta+(1\pm \varepsilon)T/4$.
  \item If no stable value is sampled on a rising clock flank, one of two cases must apply:
  \begin{enumerate}
    \item A delay of $\delta + x$ for $x\in [0,(1+\varepsilon)T/4]$ is applied and \signal{E\_OUT} behaves as if a stable $0$ was sampled.
    \item Delay $\delta$ is applied.
  \end{enumerate}
\end{enumerate}
The above claims follow from the reasoning sketched above and verifying the timing of the circuit after layout for the considered technology, provided that the pulse shaper implementation meets its specification. Provided that, for a sufficiently small $\varepsilon>0$, the pulse shaper is presented with a glitch-free signal of high time at least $(1-\varepsilon)T/2$ and low time at least $(1-3\varepsilon)T/4$, its output must satisfy the following conditions.
\begin{itemize}
  \item It reproduces each rising flank at the input with a fixed delay at the output.
  \item For each falling flank at the output, it produces a falling flank $(1\pm \varepsilon)T/2$ time after the preceding rising flank at the output.
\end{itemize}
\begin{theorem}
Suppose that dynamic changes in delays between consecutive clock cycles are negligible and the pulse shaper implementation satisfies the conditions stated above. Then the circuit given in \Cref{fig:Delay_Element_Internals} correctly implements a delay element in the sense of~\cite{DroopJournal}.
\end{theorem}
\begin{proof}[Proof Sketch]
The proof is by induction on the cycle number. On boot-up, \signal{CLK\_IN} is low for sufficiently long that all derived signals in the element settle. Now assume that we showed all properties for the first $i-1\ge 0$ cycles and consider cycle $i$. Due to the constraints on the low and high time of \signal{CLK\_IN}, the above discussion shows the first property. Moreover, regardless of the internal state of the $0$-masking latch after being latched for the $i$-th time (including metastability), the $i$-th rising clock flank is reproduced by the second \signal{NAND} gate with a delay between $\delta'$ (for some $\delta'$ with negligible variation between consecutive cycles) and $\delta'+(1+\varepsilon)T/4$. The $i$-th falling clock flank is subject to the delay $\delta'$ as well, regardless of the sampled value.

Recall the assumptions on \signal{CLK\_IN}. We get that the high time after the $i$-th rising clock flank fed into the pulse shaper is at least $(1-\varepsilon)T/2$, and the low time is at least $T-(1+\varepsilon)3T/4=(1-3\varepsilon)T/4$. Hence, the pulse shaper will guarantee the second property required from the delay element for the $i$-th pulse. Moreover, it will apply a fixed delay $\delta''$ to the rising clock flank. Thus, with $\delta:=\delta'+\delta''$, the remaining properties for the $i$-th cycle follow from the above observations regarding the timing of the rising clock flank fed into the pulse shaper depending on the sampled value and that metastability can be induced only in one of the two slave latches, with the respective other receiving the stable value defined by the masking output of the master latch.
\end{proof}

\subsection{Pulse Shaper}
\begin{figure*}[t!]
    \centering
    \includegraphics[width=0.9\textwidth]{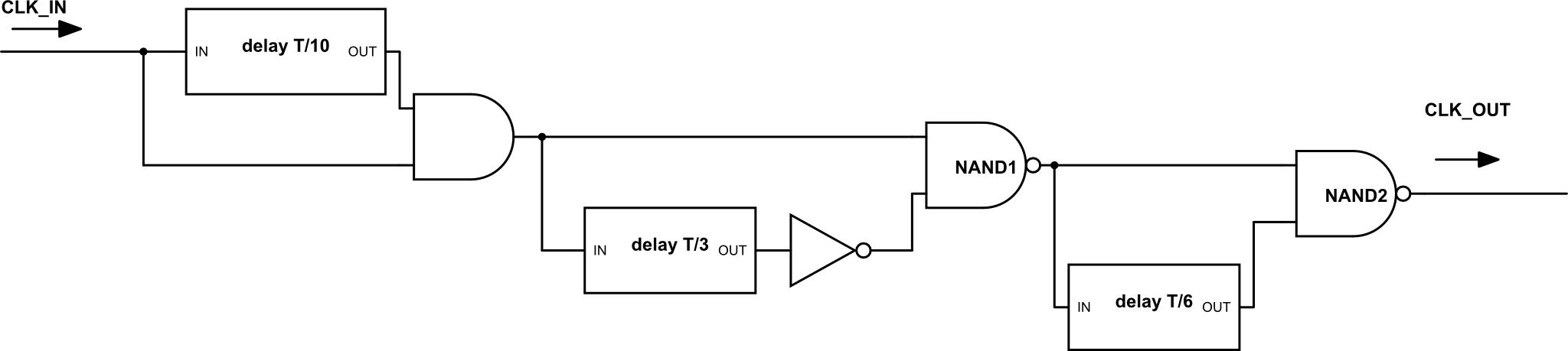}
    \caption{An idealised design of our pulse shaper. To allow for small PVT variations in the test setting we reduce the length of the delay lines \signal{delay T/3} and \signal{delay T/6} in the actual design.
    \label{fig:Pulse_Shaper}}
\end{figure*}

\begin{figure*}[t!]
    \centering
    \includegraphics[width=0.6\textwidth]{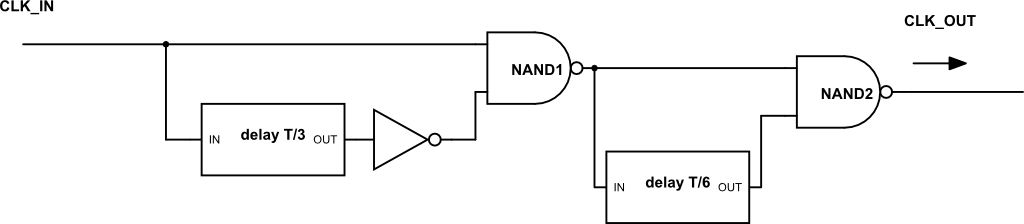}
    \caption{The Pulse Shaper of \cite{DroopJournal} for comparison with our design.
    \label{fig:Old_Pulse_Shaper}}
\end{figure*}

The pulse shaper is a submodule of the delay element. Its purpose is to ensure that the high time of the clock signal stays invariant, regardless of whether the pulse is delayed or not. The implementation proposed in~\cite{DroopJournal}, shown in \Cref{fig:Old_Pulse_Shaper}, was not proven correct. During our development process we encountered a flaw in its design. Before we discuss the fix we applied, let us analyze the design from~\cite{DroopJournal}. 

\paragraph{The Old Pulse Shaper} The schematic for the old pulse shaper is shown in \Cref{fig:Old_Pulse_Shaper}. First we briefly recap its behaviour under ideal conditions, where delays are perfect. It takes as input, a potentially uneven input clock of time period $T$, which holds \signal{logic 1} for time $xT\ (0 < x < 1)$  and low time $(1-x)T$. Provided $1/3 < x < 2/3$, it outputs an even clock which holds logic value $1$ for time $T/2$ and $0$ for logic value $T/2$. Further, for each input pulse arriving at time $t$, it provides a pulse of the output clock at time $t$. This pulse shaper design has two stages:
\begin{itemize}
    \item The first stage which begins at the input and terminates at the first NAND gate shifts and inverts the incoming clock by a fixed time delay of T/3. Then it recombines the original input clock and the delayed and inverted clock. If $x < 2/3$, the result is an output clock of time period $T$, such that for each input clock pulse arriving at time $t$, a corresponding output pulse starts at time $t$ with \signal{logic 0} for duration $T/3$ followed by \signal{logic 1} for duration $2T/3$. 
    \item The second stage takes the clock output by the first stage and combine it through a NAND gate with a version of the same clock delayed by $T/6$. The result is an output clock whose pulses start at the same time as input pulses, with \signal{logic 1} for duration $T/2$ and \signal{logic 0} for $T/2$
\end{itemize}

\paragraph{The flaw in the plan} The aforementioned design is, however, flawed on a two counts. Firstly, even in ideal design setting, when a delay element detects a droop, its pulse shaper receives clock pulses which hold \signal{logic 1} for a duration of $3T/4$. However, the pulse shaper of \cite{DroopJournal} is only designed to operate when this duration is at most $2T/3$. if $x > 2T/3$, the clock pulse will start later than $t$ and have \signal{logic 0} for a time less than $T/3$. Further, if $x > 5T/6$, then the output clock of stage 1 holds \signal{logic 1} for less than a duration of $T/6$. Thus in the second stage, the rising flank of each pulse in the delayed clock occurs after the falling flank of the undelayed clock. This creates a glitchy output clock. In an idealised delay setting, this second issue might not be of concern. However, delay variations or buffers which unevenly delay rising and falling clock flanks can violate the assumptions about the duration of \signal{logic 1} in clock cycles. Further, in a PLL free design we cannot be certain that our delay cells actually remain immune to PVT variations. Unlike the phase accumulator we cannot hope to replace our delay lines by a flip-flop based scheme here. Thus, we adopted several changes in our design. 

\begin{itemize}
    \item To address the issue of incoming pulses whose \signal{logic 1} portion might be too large, we introduce an additional phase to the pulse shaper which shortens our pulses by $T/10$. This can be seen in \Cref{fig:Pulse_Shaper} 
    \item We use specialized long-delay buffer cells in our delay line that have smaller variations in the delay of rising and falling clock flanks.
    \item We modify the delays of $T/3$ to $T/4$ and $T/6$ to $T/5$ which increases our tolerance for PVT variations.
\end{itemize} 

 Next we explain how shortening the delay lines in the two stages of the pulse shaper affects our pulse shaper's output. We recall again that in ~\cite{DroopJournal}, the value of $x$ is $T/3$, entailing that the high and low time of the signal must be at least $T/3$. However, the implementation of the delay element reduces the low time of the signal that is fed into the pulse shaper from $T/2$ to $T/4$ during nominal operation, i.e., when the clock signal is not delayed. This means that the delayed and negated input to the first \signal{NAND} gate of the pulse shaper implementation is still low when the new rising clock flank arrives. Accordingly, the \signal{NAND} gate transitions only when the delay line output transitions, i.e., (roughly) $T/3-T/4$ time after the rising clock flank arrives. Consequently, the \signal{NAND} gate has low output for only $T/4$ time, to which $T/6$ time is added for a high time of $T/4+T/6=5T/12$ at the output. Note, however, that the overall delay of the pulse shaper is affected as well, and if the delay element delays the rising clock flank, it does not change the high time of the respective clock pulse.

This behavior of the circuit from~\cite{DroopJournal} was confirmed in simulation. Moreover, if delay variations result in a high time approaching $T/3$ -- which is only by $T/12$ smaller than $5T/12$ -- we end up with glitches in the clock signal; reducing the high time further leads to lost clock pulses. During the design of our circuit, we encountered exactly such glitches in simulations, alerting us of the oversight in~\cite{DroopJournal}, where the issue was missed due to the assumption of accurate delays.

%###
\paragraph*{Our implementation}
%###
Observe that the delay line in the first stage of the pulse shaper design from~\cite{DroopJournal} limits the low time of the intermediate signal, while the second stage produces a corresponding high time extended by the delay of the second delay line. The latter must be shorter than the low time of the pulse it receives, as the output of the delay line must transition before the undelayed input signal to the second \signal{NAND} gate, to avoid glitches in the output signal. However, there is no need to extend the high time by $T/2-x$ in a single step. Instead, one can add smaller amounts that add up to $T/2-x$ in multiple stages. For example, using four stages, one could choose delays of $T/8$, $T/10$, $T/6$, and $T/2-T/6-T/8-T/10=13T/120$, respectively, resulting in the constraints
\begin{align*}
T-\frac{(1+\varepsilon)3T}{4}&>\frac{(1+\varepsilon)T}{8}\\
\frac{(1-\varepsilon)T}{8}&>\frac{(1+\varepsilon)T}{10}\\
(1-\varepsilon)T\left(\frac{1}{8}+\frac{1}{10}\right)&>\frac{(1+\varepsilon)T}{6}\\
(1-\varepsilon)T\left(\frac{1}{6}+\frac{1}{8}+\frac{1}{10}\right)&>\frac{(1+\varepsilon)13T}{120},
\end{align*}
which is satisfied for $\varepsilon<1/9$, i.e., a delay variation of more than $10\%$ could be supported.

However, even more stages would be required to account for the possibility of even higher variations. Instead, we decided to exploit that we are free to reduce the input clock frequency, because the phase accumulator can operate with any input period of at least $T$. The delay element and hence also the pulse shaper are then presented with larger high and low times of their input signals, relaxing the constraints on the pulse shaper accordingly. With this in mind, we prepended only a single stage to the pulse shaper, which shortened input pulses by roughly $T/10$, cf.~\Cref{fig:Pulse_Shaper}. This struck a balance between limiting the deterioation in the supported clock frequency and design and verification effort. Based on the above analysis, this pushed the high and low times for the signal fed into the subsequent stages within the acceptable limits, but only for an $\varepsilon$ of less than $2\%$. Therefore, in our implementation any further variation has to be compensated for by reducing the frequency at which the system operates.

\subsection{Droop Detector}

\begin{figure*}[t!]
    \centering
    \includegraphics[width=0.95\textwidth]{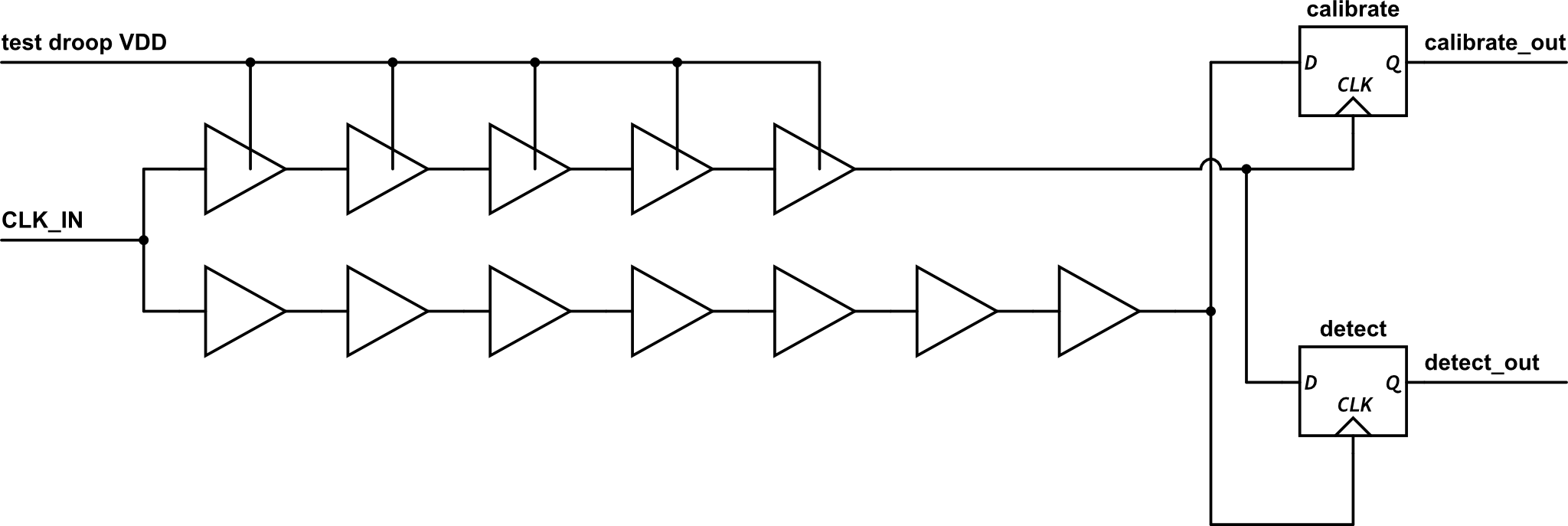}
    \caption{The Design of our Droop Detector}
    \label{fig:droop-detector}
\end{figure*}
Finally, we describe the design of the droop detector we used on our test chip. Since our circuit can operate with any detector design, we did not seek an optimized solution in our proof-of-concept circuit. Accordingly, we favored simplicity of design over performance in this subcircuit.

Our design implements a variation of the schematic presented in \Cref{fig:droop-detector}. In our design, we exploit the excess delay that a droop voltage induces in circuit elements. Our circuit consists of two delay lines, one slightly longer than the other. All buffers in the lower delay line, which we call the \emph{reference line} are connected to the chip's \signal{VDD} and \signal{GND} lines. For an $x$ determined by the process technology, this lower line contains $x+2$ buffers. The upper line, which we call the \emph{droop test line} contains $x$ buffers. Both these lines receive \signal{CLK\_IN} as common input. For the purpose of experimentation, we also use a variant of the design in which we can freely alter the \signal{VDD} inputs of all buffers on the droop test line, by connecting them to an analog voltage input named \signal{test droop VDD}. When testing our droop detector, this line then receives controlled test voltages below the standard \signal{VDD} voltage. 

Under baseline operating conditions $\signal{test droop VDD} = \signal{VDD}$. Thus the delay of each buffer cell on both the reference line and droop test line is almost identical (barring process variations). Thus the sampling edge of \signal{CLK\_IN} arrives 2 buffer delays earlier at the end of the droop test line than it does at the end of the reference line. From the perspective of the calibrate \signal{D Flip Flop}, at each sampling edge of its clock input, the data sample is \signal{logic 1}. Thus it maintains an output of \signal{logic 1}. For the detect \signal{D Flip Flop}, the data input remains \signal{logic 0} at each sampling edge of its clock input. Thus it maintains an output of \signal{logic 0}.

When the signal \signal{test droop VDD} is given a sufficiently low voltage input, i.e. a sufficiently large voltage droop, the collective delay of the droop test line exceeds the delay of the reference line. Thus the situation at the flip flops is reversed. The calibrate flip flop now outputs \signal{logic 0} and the detect flip flop outputs \signal{logic 1}.

%\section{Guarantees of our Design}\label{sec:synth_layout}
\section{Synthesis and Layout}\label{sec:synthesis}
In this section we describe our design decisions related to synthesis and layout.
\begin{comment}
\TODO{\textbf{Goran}: This is a deletable TODO note to help you find this section. We basically hope to cover the following points:
}
\begin{itemize}
    \item Choice of ports
    \item Details of synthesis : Basically how we make all library cells "Do not touch" to avoid the problem that they remove our delay lines to optimize the delay.
    \item Use of large delay cells for symmetric delays on rising and falling transitions instead of buffer chains
    \item Limitations of input signal and output signal frequencies for the IHP ports.
    \item Details about the layout : for example area is took on the chip.
    \item Any other details that you think you should be mentioned here.         
\end{itemize}
\end{comment}
\begin{figure}
    \centering
    \includegraphics[width=0.5\linewidth]{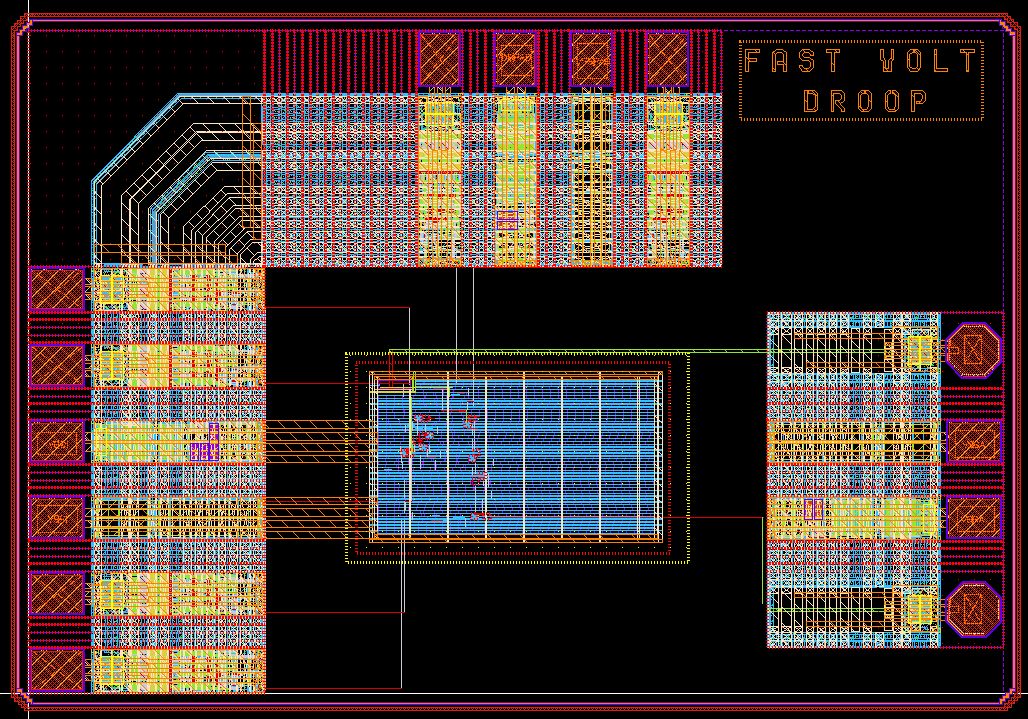}
    \caption{The layout of our chip}
    \label{fig:chip-pic-virtuoso}
\end{figure}

The target technology for implementing the chip design was a 130 nm CMOS process from IHP. Before we started implementing the entire design, we created the required library components for two masking latches (see \Cref{fig:MC_Latch_0,fig:MC_Latch_01}) as well as for the droop detector circuit (\Cref{fig:droop-detector}). This included Verilog functional models, timing models (Liberty), abstracts (LEF), schematics (Cadence), and layouts (Cadence) of all cells. The two latch cells were designed to meet the library requirements of the IHP standard cell design, allowing for automatic placement and routing during the implementation process. The droop detector circuit was designed as a hard macro and manually placed during floorplan design. Regarding the two analog signals of the droop detector circuit (VDD and Detect), we designed customized analog pads with ESD protection that allowed driving and detecting voltage levels below the standard VDD. The final design was synthesized for the target frequency of 400 MHz. During the layout design phase, we had to pay special attention to the clock buffer delays in the clock tree synthesis process, since we had to properly balance the propagation of the clock tree at the clock input of the droop detector, which already included customized delay lines in the clock net as part of the block macro. Finally, the total area of the chip was about 1 mm², most of which was occupied by the padring, which included a total of 16 I/O pads (10 signal and 6 power/ground pads). The power consumption was estimated at 4 mW. The chip photo is shown in \Cref{fig:chip-pic-virtuoso} .

%\section{Conclusion}\label{sec:conclusion}
%\input{contents/03-Conclusions.tex}

\bibliographystyle{IEEEtran}
\bibliography{conference}

\listoftodos
\end{document}